\documentclass[10pt]{article}

\usepackage[]{fullpage}
\usepackage{psfrag}
\usepackage{graphicx}
\usepackage{dsfont}
\usepackage{color}
\usepackage{url}

\usepackage{amsfonts}
\usepackage{amsthm}
\usepackage{amsmath}
\usepackage{amssymb}
\usepackage{cite}
\graphicspath{{figs/}}
\usepackage{pseudocode}

\newtheorem{theorem}{Theorem}

\newtheorem{remark}{Remark}

\title{A Counterexample to the Vector Generalization of Costa's EPI, \\and Partial Resolution}

\author{Thomas~A.~Courtade$^*$, Guangyue Han$^{\dagger}$ and Yaochen Wu$^{\dagger}$\vspace{1ex}\\
$^{*}$UC Berkeley, Department of Electrical Engineering and Computer Sciences\\
$^{\dagger}$The University of Hong Kong, Department of Mathematics
}

\begin{document}
\maketitle

\begin{abstract}
We give a counterexample to the vector generalization of Costa's entropy power inequality (EPI) due to Liu, Liu, Poor and Shamai.  In particular, the claimed inequality can fail if the matix-valued parameter in the convex combination does not commute with the covariance of the additive Gaussian noise.  Conversely,   the inequality  holds if these two matrices commute.\end{abstract}

\vspace{-3ex}

\begin{center}
\rule{0.5\textwidth}{.4pt}
\end{center}

For a random vector $X$ with density on $\mathbb{R}^n$, we let $h(X)$ denote its differential entropy.  Let $Z\sim N(0,\Sigma_Z)$ be a Gaussian vector in $\mathbb{R}^n$ independent of $X$, and let $\mathrm{A}$ be a (real symmetric) positive semidefinite $n\times n$ matrix satisfying $\mathrm{A} \preceq \mathrm{I}$ with respect to the positive semidefinite ordering, where $\mathrm{I}$ denotes the identity matrix.   In \cite[Theorem 1]{bib:GenCosta}, Liu, Liu, Poor and Shamai claim the following  generalization of Costa's EPI\footnote{All entropies are taken to be base $e$ throughout.  Also, for a  positive semidefinite matrix $\mathrm{M}$, we write $\mathrm{M}^{1/2}$ to denote the unique positive semidefinite matrix such $\mathrm{M} = \mathrm{M}^{1/2}\mathrm{M}^{1/2}$.} \cite{costa1985new}:
\begin{align}
e^{\frac{2}{n}h (X + \mathrm{A}^{1/2} Z)  } &\geq |\mathrm{I}-\mathrm{A}|^{1/n} e^{\frac{2}{n} h(X) } + |\mathrm{A}|^{1/n} e^{\frac{2}{n}h(X+ Z)}.\label{eq:costaEPI}
\end{align}
Liu et al.\ apply \eqref{eq:costaEPI} in their investigation of the secrecy capacity region of the degraded vector Gaussian broadcast channel with layered confidential messages. 

The purpose of this note is to demonstrate that \eqref{eq:costaEPI}  can fail for  $n\geq2$, and also to offer a partial resolution.  Toward the first goal, consider $n=2$ and let us define
\begin{align}
\Sigma_X = \begin{pmatrix} 200 & 100\\ 100 & 51
\end{pmatrix},  ~~~
\Sigma_Z = \begin{pmatrix} 200 & 0\\ 0 & 1
\end{pmatrix}, ~~~
\mathrm{A}^{1/2} = \frac{1}{20}\begin{pmatrix} 10 & 5\\ 5 & 17
\end{pmatrix}.\label{matrixChoice}
\end{align}
Taking $X\sim N(0,\Sigma_X)$ and $Z\sim N(0,\Sigma_Z)$ to be independent Gaussian vectors, we have
\begin{align*}
\tfrac{1}{2\pi e}e^{\frac{2}{n}h (X + \mathrm{A}^{1/2} Z)  }  = |\Sigma_X + \mathrm{A}^{1/2} \Sigma_Z \mathrm{A}^{1/2} |^{1/2} \approx 19.53. 
\end{align*}
On the other hand, 
\begin{align*}
\tfrac{1}{2\pi e}\left(  |\mathrm{I}-\mathrm{A}|^{1/n} e^{\frac{2}{n} h(X) } +  |\mathrm{A}|^{1/n} e^{\frac{2}{n}h(X+ Z)} \right)
 &=  |\mathrm{I}-\mathrm{A}|^{1/2}  |\Sigma_X|^{1/2}+|\mathrm{A}|^{1/2}   |\Sigma_X+ \Sigma_Z|^{1/2}\approx 40.28. 
\end{align*} 
Thus, a contradiction to \eqref{eq:costaEPI} is obtained.  We remark that there is nothing particularly unique about this counterexample, except that the matrices were  chosen to violate \eqref{eq:costaEPI} by a significant margin.

Evidently,  further assumptions are needed  in order for \eqref{eq:costaEPI} to hold.  As we now argue, it suffices for the matrices $\mathrm{A}$ and $\Sigma_Z$ to commute.  Fortunately, the critical application of inequality \eqref{eq:costaEPI} by Liu et al. (see \cite[p. 1877]{bib:GenCosta}) assumes only that $\mathrm{A}$ and $\Sigma_Z$ are diagonal matrices, so the main conclusions of  \cite{bib:GenCosta} appear to be unaffected, aside from \cite[Theorem 1]{bib:GenCosta}.  That said, reference \cite{bib:GenCosta} has been cited numerous times in the literature, so other published results may be impacted. We do not attempt to give an account of these consequences here, but note that this matter  deserves independent attention. 
\begin{theorem}\label{thm:FixedCosta}
Let $X$ be a random vector with density on $\mathbb{R}^n$, and let $Z\sim N(0,\Sigma_Z)$ be a Gaussian vector in $\mathbb{R}^n$ independent of $X$.  If $\mathrm{A}  \preceq \mathrm{I}$ is   positive semidefinite and commutes with $\Sigma_Z$, then 
\begin{align}
e^{\frac{2}{n}h (X + \mathrm{A}^{1/2} Z)  } &\geq |\mathrm{I}-\mathrm{A}|^{1/n} e^{\frac{2}{n} h(X) } + |\mathrm{A}|^{1/n} e^{\frac{2}{n}h(X+ Z)}.\notag
\end{align}
\end{theorem}
\begin{proof}
Since both  $\mathrm{A}$ and $\Sigma_Z$ are symmetric, they commute if and only if they are simultaneously  diagonalizable; i.e., $\Sigma_Z = \mathrm{U} \Lambda \mathrm{U}^T$ and $\mathrm{A} = \mathrm{U} \mathrm{D} \mathrm{U}^T$, for $\mathrm{U}$ orthogonal and $\Lambda, \mathrm{D}$ diagonal matrices. Since $\mathrm{A}^{1/2} = \mathrm{U} \mathrm{D}^{1/2} \mathrm{U}^T$ and $(\mathrm{I}-\mathrm{A})^{1/2}=\mathrm{U} (\mathrm{I} -\mathrm{D})^{1/2} \mathrm{U}^T$, we have the following identity
\begin{align}
\mathrm{A}^{1/2} \Sigma_Z \mathrm{A}^{1/2} + (\mathrm{I}-\mathrm{A})^{1/2}\Sigma_Z (\mathrm{I}-\mathrm{A})^{1/2} = \Sigma_Z.\label{matricesCommute}
\end{align}
It was established in \cite{courtade2016strengthening2} that if $X, Y, W$  are independent random vectors in $\mathbb{R}^n$, with $W$ Gaussian, then 
\begin{align}
e^{\frac{2}{n}h(X+ W)}e^{\frac{2}{n}h(Y+ W)} \geq e^{\frac{2}{n}h(X)}e^{\frac{2}{n}h(Y)} + e^{\frac{2}{n}h(W)}e^{\frac{2}{n}h(X+ Y+W)}. \label{tripleEPI}
\end{align}
To apply this to our setting, let $Y =  (\mathrm{I}-\mathrm{A})^{1/2}Z_1$ and $W =    \mathrm{A}^{1/2}Z_2$, where $Z_1,Z_2$ are independent copies of $Z$.  It follows from \eqref{matricesCommute} that $Y+W = Z$ in distribution, so that \eqref{tripleEPI} particularizes to
\begin{align}
e^{\frac{2}{n}h(X+ \mathrm{A}^{1/2}Z )}e^{\frac{2}{n}h(Z)} \geq  |\mathrm{I}-\mathrm{A}|^{1/n} e^{\frac{2}{n}h(X)}e^{\frac{2}{n}h(Z)} +  |\mathrm{A}|^{1/n}e^{\frac{2}{n}h(Z)}e^{\frac{2}{n}h(X+ Z)}.\notag
\end{align}
Dividing through by $e^{\frac{2}{n}h(Z)}$ completes the proof. 
\end{proof}

\begin{remark}
The proof above was originally given  in  \cite{courtade2016strengthening2} as an application demonstrating that  \eqref{eq:costaEPI} could be recovered from \eqref{tripleEPI}.  However, it was mistakenly claimed that $(\mathrm{I}-\mathrm{A})^{1/2}Z_1 +     \mathrm{A}^{1/2}Z_2$ was equal in distribution to  $Z$, which need not hold unless $\mathrm{A}$ and $\Sigma_Z$ commute.   
\end{remark}

It turns out that Liu et al.'s proof is also valid under the assumption that $\mathrm{A}$ and $\Sigma_Z$ commute. More precisely,  Liu et al.'s proof contains an incorrect application of the AM-GM  inequality in the form:
\begin{align}
|\Sigma_Z^{-1}\, \mathsf{Cov}(Z|\mathrm{D}_{\gamma}X+Z)(\mathrm{I}-\mathrm{D}_{\gamma}^{-2})|^{1/n}\leq \frac{1}{n}\mathsf{Tr}(\Sigma_Z^{-1}\, \mathsf{Cov}(Z|\mathrm{D}_{\gamma}X+Z)(\mathrm{I}-\mathrm{D}_{\gamma}^{-2})),\label{AMGM}
\end{align}
where $\mathrm{D}_{\gamma} := (\mathrm{I}+ \gamma(\mathrm{A}-\mathrm{I}))^{1/2}$, and $\gamma\in [0,1]$ parameterizes a path of perturbation.  Inequality \eqref{AMGM} can fail if the argument of the trace is not  positive semidefinite, and  a product of positive semidefinite matrices is  not necessarily positive semidefinite.  For instance, returning to the counterexample above where $X\sim N(0,\Sigma_X)$ and matrices are chosen according to \eqref{matrixChoice}, then the eigenvalues of $\Sigma_Z^{-1}\, \mathsf{Cov}(Z|\mathrm{D}_{\gamma}X+Z)(\mathrm{I}-\mathrm{D}_{\gamma}^{-2})$ can be approximately computed as   $\{-0.0053, -0.7273\}$ for $\gamma=0.5$, in violation of \eqref{AMGM}.

However,  if $\mathrm{A}$ and $\Sigma_Z$ commute, then so do $\Sigma_Z^{-1/2}$ and $(\mathrm{I}-\mathrm{D}_{\gamma}^{-2})^{1/2}$ (since both are simultaneously diagonalizable by a common orthogonal matrix $\mathrm{U}$).  Hence,   
\begin{align}
\mathsf{Tr}(\Sigma_Z^{-1}\, \mathsf{Cov}(Z|\mathrm{D}_{\gamma}X+Z)(\mathrm{I}-\mathrm{D}_{\gamma}^{-2}))
= \mathsf{Tr}((\mathrm{I}-\mathrm{D}_{\gamma}^{-2})^{1/2} \Sigma_Z^{-1/2}\, \mathsf{Cov}(Z|\mathrm{D}_{\gamma}X+Z)\Sigma_Z^{-1/2}(\mathrm{I}-\mathrm{D}_{\gamma}^{-2})^{1/2} ).\notag
\end{align}
The argument of the second trace term is clearly positive semidefinite, and therefore \eqref{AMGM} holds for all $\gamma\in[0,1]$ under the additional assumption that  $\mathrm{A}$ and $\Sigma_Z$ commute, thereby repairing Liu et al.'s proof.
 
In closing, we remark that the additional assumption that $\mathrm{A}$ and $\Sigma_Z$ commute is a relatively strong one.  It can easily be seen, using the simultaneous diagonalization property of $\mathrm{A}$ and $\Sigma_Z$ by a common orthogonal matrix $\mathrm{U}$, that Theorem \ref{thm:FixedCosta} has a completely equivalent statement where $Z\sim N(0,\mathrm{I})$ and $\mathrm{A}$ is restricted to be a diagonal matrix with diagonal entries $0 \leq a_{ii}\leq 1$, $i=1, \dots, n$.   

\end{document}